\documentclass[12pt]{article}

\usepackage{amsfonts}
\usepackage{amssymb}
\usepackage{amsmath}
\usepackage{amsthm}
\usepackage{eurosym}
\usepackage{upgreek}
\usepackage{tikz}
\usepackage{graphicx}
\usepackage{multirow}
\usepackage{subfigure}
\usepackage{pgfplots}
\usepackage{stmaryrd}
\usepackage{units}
\RequirePackage[colorlinks,citecolor=blue,linkcolor=red,urlcolor=blue]{hyperref}
\RequirePackage{natbib}

\usetikzlibrary{shapes,snakes}

\theoremstyle{theorem}
\newtheorem{theorem}{Theorem}

\newtheorem{lemma}{Lemma}
\newtheorem{exam}{Example}
\newenvironment{example}{\begin{exam} \rm }{\hfill  $\triangleleft$ \end{exam}}\newtheorem{des}{Design}

\newtheorem*{exam1}{Example 1 (continued)}

\newtheorem*{exam2}{Example 2 (continued)}

\newtheorem*{exam4}{Example 4 (continued)}
\newenvironment{example4}{\begin{exam4} \rm }{\hfill  $\triangleleft$ \end{exam4}}
\newtheorem{assum}{Assumption}

\newtheorem{rema}{Remark}
\newenvironment{remark}{\begin{rema} \rm }{\hfill $\triangleleft$ \end{rema}}
\newtheorem{defin}{Definition}
\newenvironment{definition}{\begin{defin} \rm }{\hfill  \end{defin}}
\newtheorem{mechan}{Mechanism}

\pgfplotsset{my style/.append style={axis x line=middle, axis y line=
           middle}}

\DeclareMathOperator{\supp}{supp}

\DeclareMathOperator{\inter}{int}

\DeclareMathOperator{\rank}{rank}
\DeclareMathOperator{\high}{high}
\DeclareMathOperator{\low}{low}
\DeclareMathOperator{\obj}{obj}

\DeclareFontFamily{U}{mathx}{\hyphenchar\font45}
\DeclareFontShape{U}{mathx}{m}{n}{
      <5> <6> <7> <8> <9> <10>
      <10.95> <12> <14.4> <17.28> <20.74> <24.88>
      mathx10
      }{}
\DeclareSymbolFont{mathx}{U}{mathx}{m}{n}
\DeclareMathSymbol{\bigtimes}{1}{mathx}{"91}

\makeatletter

\renewcommand*{\@seccntformat}[1]{%
  \csname the#1\endcsname.\quad}
\makeatother

\begin{document}

\title{Belief identification by proxy\footnote{I am greatly indebted to Edi Karni, Jay Lu, Andy Mackenzie and Peter Wakker for their instrumental input at various stages of this project. Also many thanks to Mohammed Abdellaoui, Jean Baccelli, Aurelien Baillon, Georgios Gerasimou, Itzhak Gilboa, Pierfrancesco Guarino, Jeanne Hagenbach, Peter Hammond, Hannes Leitgeb, Lasse Mononen, Lise Vesterlund, Sasha Vostroknutov, Roberto Weber, and the audiences at FUR (Ghent), LOFT (Groningen), Workshop on strategic communication, bounded rationality and complexity (Cergy), Workshop on strategic communication (Graz), Epicenter Workshop (Maastricht), RCEA (online) and internal seminars in Maastricht University for very useful comments. }}






\author{
\textsc{Elias Tsakas}\footnote{Department of Microeconomics and Public Economics, Maastricht University, P.O. Box 616, 6200 MD, Maastricht, The Netherlands; Homepage: \url{www.elias-tsakas.com}; E-mail: \href{mailto:e.tsakas@maastrichtuniversity.nl}{\texttt{e.tsakas@maastrichtuniversity.nl}}}\\ 
\small{\textit{Maastricht University}}}

\date{\small{November, 2023}}

\maketitle

\begin{abstract}

\noindent It is well known that individual beliefs cannot be identified using traditional choice data, unless we exogenously assume state-independent utilities. In this paper, I propose a novel methodology that solves this long-standing identification problem in a simple way. This method relies on the extending the state space by introducing a proxy, for which the agent has no stakes conditional on the original state space. The latter allows us to identify the agent's conditional beliefs about the proxy given each state realization, which in turn suffices for indirectly identifying her beliefs about the original state space. This approach is analogous to the one of instrumental variables in econometrics. Similarly to instrumental variables, the appeal of this method comes from the flexibility in selecting a proxy.

\vspace{0.5\baselineskip}



\end{abstract}

\section{Introduction}

\subsection{Background} 

Identifying subjective beliefs is a central problem in economics that dates back to  the seminal contributions of \cite{Ramsey1931}, \cite{DeFinetti1937} and \cite{Savage1954}. While this literature originally focused on providing foundations for subjective probability, more recently economists have also recognized the practical importance of the question \citep{Manski2004}. This renewed interest comes from the fact that beliefs are nowadays used for a broad range of purposes, e.g., to make out-of-sample predictions; to obtain, compare, and aggregate forecasts; to study irregularities in information processing, etc. Therefore, accurately measuring actual beliefs is of outmost importance for applied research and policy.\footnote{Throughout the paper I assume that an actual belief exists, and it is therefore interpreted as an unobservable primitive. Still, my entire analysis can be easily adapted to classical settings where the belief is interpreted as a parameter that only acquires meaning within a SEU model (Section \ref{S:decision theoretic foundation}).}


The traditional methodology for belief identification relies on observing betting behavior: an agent's choices among acts are supposed to reveal her beliefs over the state space \citep{Savage1954, AnscombeAumann1963, Wakker1989}. However, there is a foundational caveat, known as the \textit{identification problem} \citep[e.g.,][]{Dreze1961,Dreze1987,Fishburn1973, KarniSchmeidlerVind1983}: 
\begin{quote}
\textit{The agent's beliefs can be identified through traditional choices over acts, only if we exogenously assume her utility function to be state-independent.}
\end{quote}
Let me illustrate the problem by means of an example.\footnote{This is a modification of the standard example, which was originally used in a letter correspondence between Bob Aumann and Leonard Savage \citep{Dreze1987}, and has been used widely in this literature since then. We will use this story as our running example throughout the paper. For a formal presentation, see Example \ref{EX:insurance problem}.} Suppose that a man suffers from Guillain-Barr\'{e} syndrom, a serious neurological condition that has left him completely paralyzed, and it is unclear whether he will recover from the disease within the next year. Suppose that we want to identify the subjective probability that his wife attaches to him recovering. Since she loves her husband very much, she values extra money more in a world where he has recovered and she can spend it with him, compared to a world where he is paralyzed. Hence, we cannot assume her utility over monetary payoffs to be state-independent. As a result, the relationship between her beliefs and her willingness to accept a bet is confounded by her relative utility for money between the two states. And since the latter is no longer normalized to 1 ---as it would have been under the state-independence assumption--- we cannot identify beliefs from her betting behavior. 

This poses serious problems. For starters, in many economic applications where we typically elicit beliefs, utility functions are state-dependent, either due to intrinsic preferences over states or due to unobservable state-dependent side payoffs, e.g., in insurance problems \citep{Arrow1974, CookGraham1977, DrezeRustichini2004, Karni2008b}, legal judgments \citep{Andreoni1991,FeddersenPesendorfer1998, Tsakas2017}, real estate decisions \citep{CaseShiller2012}, medical decisions \citep{PaukerKassirer1975,PaukerKassirer1980,Lu2019}. Moreover, widely-studied phenomena like political polarization are typically explained by psychological theories that rely on utilities being state-dependent, e.g., motivated reasoning  \citep{Kunda1990, Benabou2015} or ingroup favoritism \citep{Everett2015}. 

Historically, the identification problem was already noticed during the very early days of decision theory, but was consciously put aside: as Leonard Savage admits in his well-known letter correspondence with Bob Aumann, ``the problem is serious, but I am willing to live with it until something better comes along" \citep{Dreze1987}. Most of the proposed solutions rely on non-traditional choice data (see literature overview, in Section \ref{S:literature}). However, none of them is unanimously accepted as the standard one, largely because non-traditional choice data are complex and cumbersome. This suggests that the identification problem is both hard-to-solve and still open.

\subsection{Contribution}

In this paper, I propose a novel solution, which is both theoretically sound and tractable. The key idea is to stick to the standard methodology of using traditional choice data, albeit over an extended state space. I do this by taking a product space, where one dimension is the original state space and the second dimension is what I call a \textit{proxy}. The crucial feature is that the agent does not have any stakes in the realization of the proxy conditional on our original state space. This allows us to identify her conditional beliefs about the proxy using standard elicitation mechanisms. The main result leverages these conditional beliefs about the proxy to uniquely identify the actual belief about the original state space (Theorem \ref{T:Main Theorem}). 

Let me illustrate the identification result in the context of the running example. Suppose that there is a promising new experimental drug in the market which is believed to expedite the recovery time from Guillain-Barr\'{e}. The husband is eligible to participate in the last phase of the clinical trial. A proxy describes the two possible contingencies regarding his participation, i.e., he will receive either the drug $(t_1)$ or a placebo $(t_2)$. The wife is informed that his chances of being placed in the treatment group are 50\%, but she is not told in which group he is actually placed in the end. Of course, if she were to learn that he received the drug, her belief would change to some $\nu:=\pi_S(\cdot|t_1)$; on the other hand, if she learned that he received the placebo, her belief would remain unchanged at $\mu:=\pi_S(\cdot|t_2)$.\footnote{Here, I implicitly assume that the husband does not even know that he is participating in the clinical trial, and therefore his health cannot be affected by psychological factors, such as the well-known placebo effect.} For a graphical illustration, see Figure \ref{FIG:proxies graphical representation} below.

\begin{figure}[h!]
\begin{center}
\begin{tikzpicture}[scale=1.2]
\draw[->] (0,1) -- (1.8,0);
\draw[->] (0,1) -- (1.8,2);
\draw (2.3,0) node[left] {\footnotesize{$t_2$}};
\draw (2.3,2) node[left] {\footnotesize{$t_1$}};
\draw[->] (2.3,0) -- (4.3,0);
\draw[->] (2.3,0) -- (4.3,1.9);
\draw[->] (2.3,2) -- (4.3,0.1);
\draw[->] (2.3,2) -- (4.3,2);
\draw (4.3,0) node[right] {\footnotesize{$s_2:$ the husband remains paralyzed}};
\draw (4.3,2) node[right] {\footnotesize{$s_1:$ the husband recovers}};
\draw (0.7,1.8) node {\footnotesize{$50\%$}};
\draw (0.7,0.2) node {\footnotesize{$50\%$}};
\draw (3,2) node[above] {\footnotesize{$\nu_1$}};
\draw (3,0) node[below] {\footnotesize{$\mu_2$}};
\draw (2.6,0.6) node {\footnotesize{$\mu_1$}};
\draw (2.6,1.4) node {\footnotesize{$\nu_2$}};
\end{tikzpicture}
\end{center}
\vspace{-1\baselineskip}
\caption{\footnotesize{The two contingencies that the proxy describes are ``the husband takes the experimental drug $(t_1)$" and ``the husband takes the placebo $(t_2)$". The prior probability of the husband (randomly) placed in the control group is 50\%. Conditional on the placebo being taken, the wife's belief remains at $\mu$. Conditional on the drug being taken, her belief changes to $\nu$.}}
\label{FIG:proxies graphical representation}
\end{figure}
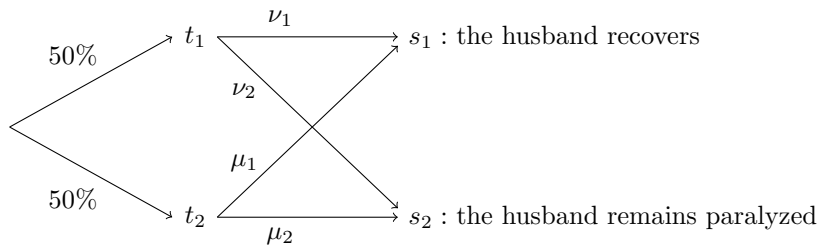

The crucial feature is that the wife \textit{does not have any stakes in the realization of the proxy conditional on her husband's health status}. For instance, conditional on him recovering (resp., remaining paralyzed), she will not care about whether he has taken the drug or the placebo, and therefore her utility from money will not depend on the realization of the proxy.\footnote{The implicit assumption that I make for the sake of this simple example is that the drug is known not have any side-effects that would make the wife care about whether her husband has taken it or not conditional on eventual health condition.} Hence, we can identify her conditional beliefs over $T$ given that her husband has recovered (viz., $\pi_T(\cdot|s_1)$), as well her conditional beliefs over $T$ given that he has remained paralyzed (viz., $\pi_T(\cdot|s_2)$). In both cases, we can do this, using standard elicitation tasks, \textit{without worrying about the identification problem}. 

Then, as long as $S$ and $T$ are not independent, we can uniquely identify her conditional belief $\pi_S(\cdot|t_2)$ given that he has taken the placebo, which is the belief $\mu$ she would have held if we had not introduced the proxy. A generalization of this simple example to any finite state space and any suitable proxy constitutes my main identification result (Theorem \ref{T:Main Theorem}).

Notably, my identification approach bears a striking conceptual similarity with the use of instrumental variables in econometrics. Recall that in IV regression we replace one awkward exogenous assumption (viz., orthogonality) with another exogenous assumption (viz., the exclusion restriction) which is easier to justify. In the same way, instead of exogenously assuming that the wife does not care about her husband's health status, we assume that she does not care about his treatment conditional on his eventual health status, with the latter being much easier to defend than the former. Thus, in both settings, it is not that we stop imposing exogenous assumptions all together, but rather that we strategically select the domain (viz., the instrumental variable and the proxy, respectively) so that the exogenous assumptions can be justified on the basis of common sense and/or existing literature. This flexibility to choose the instrumental variable is what has made IV regression so popular. Likewise, the flexibility in choosing a proxy from a large pool of of potential candidates, is what makes belief identification by proxy an appealing method.

\subsection{Literature}\label{S:literature}

The existing literature is roughly split in two streams, one that aims at providing tools for belief elicitation purposes, and an axiomatic one that focuses on achieving identification theoretically. A major conceptual difference between these two streams is that the former explicitly assumes that beliefs are an actual primitive (similarly to this paper), whereas the latter treats beliefs as a parameter within a model without much concern on whether these are actual beliefs or not. For an overview, see \cite{DrezeRustichini2004}, \cite{GrantZandt2008}, \cite{Karni2008, Karni2014} and \cite{Baccelli2017}. 

Starting with the first stream, the only papers that introduce mechanisms for eliciting beliefs under state-dependent preferences are \cite{Karni1999} and \cite{JaffrayKarni1999}, with the latter proposing two different mechanisms. In particular, \cite{Karni1999} and the first mechanism of \cite{JaffrayKarni1999} rely on assuming bounded state utilities, and they approximate the actual beliefs as monetary incentives grow arbitrarily large. The problem is that this approach requires either extreme costs, or hypothetical data. These pitfalls are recognized by the authors, who point out that in those early days there was no other option \citep[e.g.,][p.485]{Karni1999}. The second mechanism in \cite{JaffrayKarni1999} assumes state-dependence in the form of unobserved state-dependent payments: first it proceeds to elicit these payments, and subsequently to elicit beliefs using standard techniques. Unfortunately, this is a rather restrictive setting: in many applications, preferences over states are intrinsic. Moreover, eliciting the state-dependent payments is data-demanding.

Turning to the second stream, the various attempts within axiomatic decision theory differ in terms of the non-traditional choice domain they consider. There are three main methodological approaches, all of which rely on the agent's beliefs being somehow revised at some instance. 

The first such method, which was mainly followed in the early days, was based on exogenously manipulating context and a fortiori the agent's beliefs. For instance, \cite{Fishburn1973} allows for comparison between acts conditional on different events. \cite{KarniSchmeidlerVind1983}, \cite{Hammond1999} and \cite{KarniSchmeidler2016} introduce hypothetical preferences over acts conditional on exogenously given probabilities over the states. \cite{Kadane1990} allow the agent to compare lotteries at different states. \cite{Karni1992, Karni1993} allows the analyst to observe preferences conditional on different events. In some way, all of these papers assume the analyst to observe the agent's hypothetical choices. This feature makes the implementation of these methods cumbersome.

The second method relies on the idea that the agent herself can affect the state realization, and is called the \textit{moral hazard approach}. It originally appeared in the early sixties \citep{Dreze1961}, before resurfacing \citep{Dreze1987, DrezeRustichini1999} and recently receiving attention again \citep{Baccelli2021}. This literature relates to my work on a high level. The difference is that in moral hazard the analyst relies on the agent influencing the state realization, whereas in my case it is sometimes the analyst who may choose to influence the state realization (Example \ref{EX:experimental drug}). The second major difference is that the moral hazard methodology applies only to settings where the state has not been realized yet, while my method does not pose any such restriction, i.e., my theory also applies to factual beliefs. 

The third method is more recent, and relies on the agent updating her beliefs using information that the analyst provides \citep{Lu2019}. This is admittedly a very promising method, similar in spirit to some of my proxies (Example \ref{EX:expert charlatan I}). Similarly to my work, it can be potentially applied both in cases where the state has been already realized, as well as in cases where it has not.  It does not impose restrictive assumptions at the outset. The only potential drawback is that it requires stochastic choices (under different information structures), implying that a large dataset might be needed. 

Between the second and the third method, one can place a sequence of papers that rely on the agent influencing the state realization and choosing an act, conditional on different signals \citep{Karni2011a,Karni2011b,Karni2014}. The common element across these papers of Karni, the paper of Lu, and my work is that all can be thought to rely on traditional choice data over an extended state space. Nevertheless, the way this general idea is implemented is very different.

Finally, there is related work on the identification problem within non-expected utility models: While the problem is resolved in some cases outside SEU \citep{ChewWang2020, Mononen2023}, in other cases it persists \citep{Karni2020}.

\subsection{Structure of the paper}

Section \ref{S:identification problem} presents the background and formally introduces the identification problem. Section \ref{S:solution by proxy} presents my solution by proxy and formally states my main identification result. Section \ref{S:applications} leverages my identification result to provide a well-founded definition of actual utility function. Section \ref{S:decision theoretic foundation} provides decision-theoretic foundations for a proxy, thus formally distinguishing between exogenous and testable assumptions. Section \ref{S:conclusion} concludes. All proofs are relegated to the Appendix.

\section{The identification problem formalized}\label{S:identification problem}


Consider a finite state space $S=\{s_1,\dots,s_K\}$. A (female) agent has a full-support belief $\mu\in\Delta(S)$, called the \textit{actual belief}, which we want to identify. The usual choice domain is the set of acts, $\mathcal{F}_S$. An act is a function $f:S\rightarrow Q$ that maps each state $s\in S$ to a consequence $f_s:=f(s)$ in a convex subset $Q$ of a finitely dimensional Euclidean space. The agent is assumed to have preferences $\trianglerighteq$ over $\mathcal{F}_S$ that admit a \textit{State-Dependent Subjective Expected Utility (abbrev., SEU)} representation. That is, there exists some state-dependent utility function $u:Q\rightarrow\mathbb{R}^S$, such that the pair $(u,\mu)$ represents $\trianglerighteq$, i.e., for any two acts $f,g\in\mathcal{F}_S$,
\begin{equation}\label{EQ:state-dependent SEU representation}
f\trianglerighteq g \ \Leftrightarrow \ \mathbb{E}_\mu(u(f))\geq \mathbb{E}_\mu(u(g)),
\end{equation}
where 
\begin{equation}\label{EQ:SDSEU}
\mathbb{E}_{\mu}(u(f)):=\sum_{s\in S} \mu(s) u_s(f_s)
\end{equation}
is the Subjective Expected Utility of act $f$.

The appeal of SEU is that it allows us to disentangle beliefs from utilities. Is this enough for identifying the agent's actual beliefs? Unfortunately, not! The reason is that the pair $(u,\mu)$ is not the only SEU representation. Namely, take any other full-support belief $\tilde{\mu}\in\Delta(S)$, and define the rescaled utility function $\tilde{u}:Q\rightarrow\mathbb{R}^S$ so that, for every $s\in S$,
\begin{equation}\label{EQ:utility rescaling}
\tilde{u}_s=\frac{\mu(s)}{\tilde{\mu}(s)}u_s.
\end{equation}
Then, for every $f\in\mathcal{F}_S$, we obtain 
\begin{equation}\label{EQ:identification problem}
\mathbb{E}_{\tilde{\mu}}(\tilde{u}(f))=\mathbb{E}_\mu(u(f)). 
\end{equation}
Hence, the pair $(\tilde{u},\tilde{\mu})$ constitutes an alternative SEU representation of the same preferences. So, even if we observe the complete preference relation, we cannot tell if the actual belief is $\mu$ or $\tilde{\mu}$, i.e., beliefs cannot be identified from choices over $\mathcal{F}_S$.  This is known as the  \textit{identification problem} of SEU.

The early solution to the identification problem was to essentially assume it away. This approach boiled down to selecting a \textit{State-Independent Subjective Expected Utility (abbrev., SI-SEU)} representation \citep{Savage1954, AnscombeAumann1963,Wakker1989} among the infinitely many SEU representations. Formally, the idea is to pick a SEU representation $(\bar{u},\bar{\mu})$ such that $\bar{u}_s=\bar{u}_{s'}$ for all $s,s'\in S$, and then label $\bar{\mu}$ as the actual belief. This is known as the \textit{assumption of state-independent utilities}.

\begin{example}\label{EX:insurance problem}\textsc{(The wife's problem).}
Recall the introductory example, in which a man suffers from Guillain-Barr\'{e} syndrom, and it is unclear whether he will recover (state $s_1$) or remain paralyzed (state $s_2$). Suppose that his wife is a risk-neutral SEU maximizer, i.e., she evaluates monetary payoffs using the state-utility functions
$$u_1(q)=\gamma_1 q \mbox{ and } u_2(q)=\gamma_2 q,$$ 
with $\gamma_1,\gamma_2>0$. So, when we observe her being indifferent between buying and not buying an insurance with payout of \$100k and premium of \$10k, we conclude that her beliefs satisfy 
\begin{equation}\label{EQ:Example husband}
\frac{\mu(s_2)}{\mu(s_1)}=\frac{\gamma_1}{9\gamma_2}.
\end{equation}
Obviously, in order to identify her actual belief, we must specify the parameters $\gamma_1$ and $\gamma_2$. State-independent utilities correspond to one such specification, which sets $\gamma_1=\gamma_2$, and subsequently yields $\bar{\mu}(s_1)=0.90$. However, this is just one of the infinitely many parameter specifications that satisfy Equation (\ref{EQ:Example husband}), and in this sense it is an arbitrary choice.
\end{example}

To complicate things even further, oftentimes a SI-SEU representation does not even exist in the first place. This is for instance the case when state-preferences over consequences differ across states, either in some inherent way (Example \ref{EX:insurance problem ranking}), or because of unobservable state-contingent side payments (Example \ref{EX:insurance problem risk seeking}).

\begin{example}\label{EX:insurance problem ranking}
Continuing with the same example, suppose that the wife's preferences were instead represented by the SEU representation $(u,\mu)$, where
$$u_1(q)=q^2 \mbox{ and } u_2(q)=q.$$
The idea is that the wife is risk-seeking whenever her husband is healthy. In order for a SI-SEU representation to exist, it must necessarily be the case that $u_1=\alpha+\beta u_2$ for some $\alpha\in\mathbb{R}$ and $\beta>0$. It is not difficult to verify that such parameters do not exist.
 \end{example}

\begin{example}\label{EX:insurance problem risk seeking}
In the running example, suppose that the wife's preferences were instead represented by the SEU representation $(u,\mu)$, where
$$u_1(q)=(c+q)^2 \mbox{ and } u_2(q)=q^2.$$
The interpretation is that the wife is risk-seeking with quadratic state-utility function, and has an unobservable side payoff at $s_1$. Once again, there are no parameters $\alpha\in\mathbb{R}$ and $\beta>0$ such that $u_1=\alpha+\beta u_2$. So again, there is no SI-SEU representation.
\end{example}

Whenever there is no SI-SEU representation, the identification problem becomes immediately visible, as there is no obvious candidate among the infinitely many SEU representations to exogenously assume. However, this does not mean that in the converse scenario where a SI-SEU representation exists, the identification problem is resolved. In fact, the existence of a SI-SEU representation and the identification problem are conceptually orthogonal: the former is about axioms of choice being satisfied, while the latter is about model specification given that these axioms are satisfied. So, why is there a widespread misconception that they are linked? The reason is twofold and can be traced back into the history of SEU.

First of all, in the early days of decision theory, the literature was not interested in identifying actual beliefs. Instead, the focus was on predicting choices among acts, and on establishing that there exists a well-founded definition of subjective probability. And since SI-SEU representations served both purposes, the theoretical literature started treating $(\bar{u},\bar{\mu})$ as the default model, whenever it existed. Thus, when the applied literature eventually started being interested in identifying actual beliefs, it simply adopted $\bar{\mu}$ as `the correct definition of actual belief'. This is highlighted in two facts: first, (almost) the entire belief elicitation literature consists of mechanisms that elicit $\bar{\mu}$ \citep[][and references therein]{SchlagWeele2013,SchotterTrevino2014}; and second, the entire literature on State-Dependent SEU is motivated by the non-existence of SI-SEU representation \citep[][and references therein]{DrezeRustichini2004, GrantZandt2008, Karni2008, Karni2014,Baccelli2017}. 


The second source of the misconception comes from the fact that state-independent utilities are often confused terminologically with state-independent preferences, with the latter being typically captured by axioms such as the state-monotonicity axiom in \cite{AnscombeAumann1963}, or P3-P4 in \cite{Savage1954}, or state-independent preference intensity in \cite{Wakker1989}. However, state-independent preferences are weaker than state-independent utilities, as the former simply guarantees the existence of a SI-SEU representation, whereas the latter selects a SI-SEU as the `correct model' whenever it exists. 

But how big of a problem is it to assume state-independent utilities (assuming of course that a SI-SEU representation exists)? The answer depends on the purpose of belief identification. If we simply want to predict the wife's insurance choices in Example \ref{EX:insurance problem}, it is no big deal to assume state-independent utilities, as all SEU models (including SI-SEU) deliver the exact same predictions. If, on the other hand, we want to use the wife's beliefs for other purposes (e.g., as data in an experiment that studies motivated beliefs, or for predicting out-of-sample if she will book a vacation for next summer), we need to select a SEU representation that involves her actual belief. In this last case, it is reasonable to assume state-independent utilities only if we are confident that \textit{the agent has no stakes} over the state realization. 

\begin{remark}\label{REM: no stakes}
Throughout the paper, \textit{``having no stakes in the state realization"} is used synonymously to \textit{``the utility function being state-independent"}. In this sense, it is an exogenous assumption, which cannot be expressed within the language of the preferences over acts. Instead, it is simply understood as a statement about her preferences over states, which cannot be revealed by traditional choices. 
\end{remark}

\section{Main identification result}\label{S:solution by proxy}


In this paper, I propose a novel approach that solves the identification problem in a theoretically sound and practically tractable way. The idea is to continue using traditional choice data, albeit over an appropriately extended state space. Formally, I introduce another state space $T=\{t_1,\dots,t_N\}$, and define the product space $S\times T$. The agent's actual belief over the extended state space is denoted by $\pi\in\Delta(S\times T)$. For any nonempty $A\subseteq S$ and $E\subseteq T$, we respectively denote the marginal conditional beliefs $\pi_T(E|A):=\pi(A\times E|A\times T)$ and $\pi_S(A|E):=\pi(A\times E|S\times E)$. Marginal (unconditional) beliefs are simply denoted by $\pi_T$ and $\pi_S$ respectively.

\begin{definition}\label{DEF:proxy}\textsc{(Proxy).}
We say that $T$ is a proxy for $S$ whenever the following are satisfied:
\begin{itemize}
\item[$(P_0)$] \textsc{Conditional no stakes:} The agent has no stakes in $T$ conditional on the realization of $S$.
\item[$(P_1)$] \textsc{Uninformative event:} There is some event $E\subseteq T$ such that $\pi_S(\cdot|E)=\mu$.
\item[$(P_2)$] \textsc{Objective marginal belief:} $\pi_T$ agrees with an objective probability measure $\pi_T^{\obj}$.
\item[$(P_3)$] \textsc{Linear independence:} $\pi_T(\cdot|s_1),\dots,\pi_T(\cdot|s_K)$ are linearly independent in $\mathbb{R}^T$.
\end{itemize}
\end{definition}

Condition $(P_0)$ is the key feature of $T$. Formally, this is satisfied if and only if any incentive-compatible mechanism (e.g., a binarized scoring rule) will elicit the agent's actual conditional beliefs $\pi_T(\cdot|s)$ for all $s\in S$ via the strategy method. Of course, ex ante the agent may still care about the realization of $T$, due to the correlation with $S$. But once uncertainty about $S$ has been resolved, she has no longer any stakes in the realization of the proxy.

Condition $(P_1)$ states that conditioning with respect to $E\subseteq T$ does not provide any information about $S$ in comparison to the benchmark case where $T$ is not introduced in the first place. Thus, $E$ is called the uninformative event. Notice that the event $E$ is not necessarily equal to $T$ itself, meaning that the actual belief $\mu$ does not necessarily coincide with the unconditional marginal $\pi_S$ (Examples \ref{EX:experimental drug} and \ref{EX:expert charlatan I} below). 

Condition $(P_2)$ consists of two parts. First, uncertainty about $T$ is described by an objective marginal belief, $\pi_T^{\obj}$, which is formed based on commonly known facts; second, it is assumed that the agent's actual marginal belief $\pi_T$ agrees with the objective belief. This assumption often reflects the idea that $\pi_T^{\obj}$ has been publicly announced by the analyst (Examples \ref{EX:experimental drug} and \ref{EX:expert charlatan I} below), while in other cases it describes some demographic characteristic whose distribution is already commonly known to be $\pi_T^{\obj}$ (Example \ref{EX:pharmaceutical's problem} below).  

Condition $(P_3)$ is perhaps the least-obvious of the four conditions. Loosely speaking, it guarantees that there is no redundant information within $T$. Whenever $S$ is binary, this condition reduces to $\pi_T(\cdot|s_1)\neq\pi_T(\cdot|s_2)$, i.e. $S$ and $T$ are not independent. 

Note that $(P_0)-(P_2)$ are exogenous assumptions, in the sense that, although they can be sometimes falsified, none of them can be verified using traditional choice data. On the other hand, under the assumption that $(P_0)$ holds, $(P_3)$ can be tested with traditional choice data. The formal distinction between exogenous and testable conditions is made in Section \ref{S:decision theoretic foundation}, where I provide decision-theoretic foundations. For the time being, let me illustrate different types of proxies with a series of examples.  

\begin{example}\label{EX:experimental drug}\textsc{(New drug).}
Recall the wife's problem, and similarly to the introduction suppose that there is a promising new experimental drug which supposedly helps recovery. The husband is eligible to participate in a clinical trial. The proxy describes the two possible groups in which the husband can be placed:
$$T=\{\mbox{treatment group }(t_1), \mbox{ control group }(t_2)\}.$$
Then, I argue that $(P_0)-(P_3)$ satisfy common sense:
\vspace{-0.3\baselineskip}
\begin{itemize}\setlength\itemsep{-0.25em}
\item[$P_0:$] The wife will not care if he has taken the drug or the placebo, once the husband's health condition is known.
\item[$P_1:$] The wife believes that in case he receives the placebo, his chance of recovery will remain unchanged, i.e., $\pi_S(\cdot|t_2)=\mu$. Hence, $\{t_2\}$ is the uninformative event.
\item[$P_2:$] The wife knows that the chance of him being included in the treatment group is 50\%.
\item[$P_3:$] The wife believes that the drug will affect his recovery probability, i.e., $\pi_T(\cdot|s_1)\neq \pi_T(\cdot|s_2)$.
\end{itemize}
\vspace{-0.3\baselineskip}
Such proxies are called \textit{influential actions} \citep{Tsakas2020}, as the analyst (stochastically) influences the state realization via the proxy.
\end{example}

\begin{example}\label{EX:expert charlatan I}\textsc{(Expert versus charlatan).}
Once again, we consider the wife's problem, supposing now that the wife is told that her husband's file was examined by some doctor who subsequently predicted that he will recover. The proxy describes the two possible expertise levels of the doctor:
$$T=\{\mbox{expert }(t_1), \mbox{ charlatan }(t_2)\}.$$
Again, I argue that $(P_0)-(P_3)$ are reasonable conditions:
\vspace{-0.3\baselineskip}
\begin{itemize}\setlength\itemsep{-0.25em}
\item[$P_0:$] Conditional on the husband's health condition, the wife will not care if the diagnosis has come from the expert or the charlatan.
\item[$P_1:$] Conditional on the diagnosis coming from a charlatan, the wife does not revise her beliefs, i.e., $\pi_S(\cdot|t_2)=\mu$. Hence, $\{t_2\}$ is the uninformative event.
\item[$P_2:$] The wife knows that the probability of the doctor being an expert is 50\%.
\item[$P_3:$] The wife believes that the probability that the good news actually came from the expert is larger if the husband recovers than if he does not, i.e., $\pi_T(\cdot|s_1)\neq\pi_T(\cdot|s_2)$.
\end{itemize}
\vspace{-0.3\baselineskip}
Such proxies are called \textit{evidence with stochastic reliability}. A similar idea has been used in a different context by \cite{Thaler2020}.
\end{example}

\begin{example}\label{EX:pharmaceutical's problem}\textsc{(Relevant gene).} 
We consider the wife's problem for a third time. It is commonly known that recovery is correlated with the presence of specific gene that 50\% of all males carry. The gene has no other consequence. The wife does not know if her husband carries it or not. So, the proxy becomes:
$$T=\{\mbox{gene }(t_1), \mbox{ no gene }(t_2)\}.$$
Once again, I argue that $(P_0)-(P_3)$ satisfy common sense:
\vspace{-0.3\baselineskip}
\begin{itemize}\setlength\itemsep{-0.25em}
\item[$P_0:$] Conditional on the husband's health condition, the wife does not care if the husband has the gene.
\item[$P_1:$] Her actual belief is the one she holds before learning whether he has the gene or not, i.e., $\mu=\pi_T(t_1)\pi_S(\cdot|t_1)+\pi_T(t_2)\pi_S(\cdot|t_2)$. Hence, $T$ is the uninformative event.
\item[$P_2:$] The wife knows that the probability of the husband having the gene is 50\%.
\item[$P_3:$] The wife knows that recovery is correlated with the gene, i.e., $\pi_T(\cdot|s_1)\neq \pi_T(\cdot|s_2)$.
\end{itemize}
\vspace{-0.3\baselineskip}
Such proxies are induced by a \textit{demographic characteristic} whose distribution is commonly known.
\end{example}

Then, my main result shows that we can leverage conditional beliefs about the proxy to identify the actual belief about $S$.

\begin{theorem}[Main identification result]\label{T:Main Theorem}
Assume that $T$ satisfies $(P_0)-(P_2)$. Then, the actual belief $\mu\in\Delta(S)$ is uniquely identified if and only if $(P_3)$ is satisfied.
\end{theorem}

The crucial step towards identification is to first pin down the joint belief $\pi$, and subsequently to condition with respect to the uninformative event $E\subseteq T$. Let us provide an illustration.

\begin{example4}
Suppose that the wife reports probability $\pi_T(t_1|s_1)=0.80$ of the husband receiving the drug conditional on recovering, and probability $\pi_T(t_1|s_2)=0.40$ of the husband receiving the drug conditional on remaining paralyzed. There is a unique convex combination of the two aforementioned conditional probabilities inducing the objectively known $\pi_T(t_1)=0.50$, i.e., we have $0.25\pi_T(\cdot|s_1)+0.75\pi_T(\cdot|s_2)=\pi_T$. So, it is necessarily the case that $\pi_S(s_1)=0.25$, meaning that the joint belief $\pi$ is given by the following table:
\begin{center}
\begin{tikzpicture}[scale=1.9]
\draw[line width=0.8pt] (0,1) -- (0,2) -- (2.5,2) -- (2.5,1) -- (0,1);
\draw (0,1.3) node[left] {\footnotesize{placebo $(t_2)$}};
\draw (0,1.7) node[left] {\footnotesize{drug $(t_1)$}};
\draw (0,2) node[above right] {\footnotesize{recovers $(s_1)$}};
\draw (1.25,2) node[above right] {\footnotesize{paralyzed $(s_2)$}};
\draw (0.6,1.7) node {\footnotesize{$0.20$}};
\draw (1.9,1.7) node {\footnotesize{$0.30$}};
\draw (0.6,1.3) node {\footnotesize{$0.05$}};
\draw (1.9,1.3) node {\footnotesize{$0.45$}};
\end{tikzpicture}
\end{center}
Hence, her actual belief attaches probability $\mu(s_1)=\pi_S(s_1|t_2)=0.10$ to her husband recovering. Note that the identification would not have been possible if $\pi_T(\cdot|s_1)=\pi_T(\cdot|s_2)$ (i.e., if $S$ and $T$ were independent), as in this case we would have obtained $\lambda\pi_T(\cdot|s_1)+(1-\lambda)\pi_T(\cdot|s_2)=\pi_T$ for every $\lambda\in(0,1)$. 
\end{example4}

Importantly, for my identification strategy to work, it must be the case that $|\supp(\pi_T)|\geq|S|$. Otherwise, $(P_3)$ will be directly violated, since $\pi_T(\cdot|s_1),\dots,\pi_T(\cdot|s_K)$ will not be linearly independent, and a fortiori we will not be able to identify the joint belief $\pi$. The good news is that the amount of data that we need is still quite small, especially in comparison with other identification methods, e.g., whenever the state space is binary, we only need to elicit two probabilities.

On a high level, my approach bears a striking conceptual similarity to the use of instrumental variables in regression analysis. In particular, instead of automatically imposing an exogenous assumption (viz., orthogonality) on their variable of interest (viz., explanatory variable), econometricians assume it on another carefully-chosen variable (viz., instrumental variable). Similarly, instead of exogenously assuming without much thought that the agent has no stakes about the original state space (viz., $S$), I exogenously assume $(P_0)-(P_2)$ about a carefully-chosen proxy (viz., $T$). That is, in both cases, it is not that we stop imposing exogenous assumptions all together, but rather we only do so in specifically-chosen domains where the exogenous assumptions can be justified on the basis of common sense and/or well-established theoretical insights.

Such flexibility is what has made instrumental variables so successful. Similarly, I argue that my belief identification approach via proxies is appealing because there exists abundance of potentially suitable proxies, one of which will eventually satisfy common sense and/or will be consistent with well-established insights. For instance, in my running story, I have already considered three potential proxies (Examples \ref{EX:experimental drug}-\ref{EX:pharmaceutical's problem}) from which I can pick a suitable one. And of course, there are many more such proxies. This provides us with great flexibility and confidence that in practice we can almost always identify the agent's belief with my approach.

\section{Actual utility function}\label{S:applications}

In this section, I will show how my main identification result allows me to obtain a well-founded definition of an \textit{actual (state-dependent) utility function}. This is a very useful object, as it provides meaning to statements like ``a dollar is twice as valuable at $s_1$ compared to $s_2$" or ``the utility of state $s$ is equal $w(s)$", which are in turn commonly used in many applied settings, e.g., in the literature on motivated beliefs.

Note that these results will not be exclusive to my setting, and will also apply to other SEU models that uniquely identify actual beliefs. However, surprisingly the point has not been made so far in the literature. I speculate that this is because most of the existing literature does not make the explicit distinction between ``actual beliefs" and ``beliefs that are uniquely identified within a SEU model". This distinction is necessary in order to identify ``actual utilities" and is further discussed in Section \ref{S:decision theoretic foundation}.

Back in Example \ref{EX:insurance problem}, recall that every SEU representation $(\tilde{u},\tilde{\mu})$ of the wife's preferences satisfies 
$$\frac{\tilde{\mu}(s_2)}{\tilde{\mu}(s_1)}=\frac{\tilde{\gamma}_1}{9\tilde{\gamma}_2},$$
where $\tilde{u}_1(q)=\tilde{\gamma}_1 q$ and $\tilde{u}_2(q)=\tilde{\gamma}_2 q$. Therefore, there is a duality between identification of beliefs and identification of relative marginal utilities. This exact duality is what the literature traditionally leverages to go from state-independent utilities to a unique belief, viz., if we assume $\gamma_1=\gamma_2$ then we can uniquely identify $\bar{\mu}(s_1)=90\%$. In this case, on the other hand, we will use this duality the other way around, viz., once we have identified the actual belief (using Theorem \ref{T:Main Theorem}), we can uniquely identify a set of parameters $(\gamma_1,\gamma_2)$ that will constitute the ``actual utility function". For instance, in our running example, the actual belief is $\mu(s_1)=0.10$ (Example \ref{EX:experimental drug}, continued), and therefore the actual utility function can be written as
\begin{equation}\label{EQ:actual utilities}
u_1(q)=\underbrace{81\beta}_{\gamma_1} q \mbox{ and }u_2(q)=\underbrace{\beta}_{\gamma_2} q, 
\end{equation}
where $\beta>0$. The following theorem generalizes this result beyond my simple example, even in cases where there is no SI-SEU representation.

\begin{theorem}\label{THM:actual utility function}
Let $(\tilde{u},\tilde{\mu})$ be an arbitrary SEU representation. Moreover, assume that the agent's actual belief $\mu\in\Delta(S)$ has been identified (via Theorem \ref{T:Main Theorem}). Then, the class of utility functions $u:Q\rightarrow\mathbb{R}^S$ for which $(u,\mu)$ is a SEU representation is characterized by
\begin{equation}\label{EQ:identification of actual utility function}
u_s=\alpha_s+\beta\frac{\tilde{\mu}(s)}{\mu(s)}\tilde{u}_s,
\end{equation}
with $\alpha_s\in\mathbb{R}$ and $\beta>0$, for all $s\in S$.  
\end{theorem}

Not surprisingly, as this is a Bayesian framework, the arrival of new information affects the actual beliefs, but not the actual tastes. So, the set of actual utility functions remains invariant when the agent updates her beliefs upon observing new information about $S$.

\begin{remark}
Interestingly, if the preferences admit a SI-SEU representation $(\bar{u},\bar{\mu})$, Theorem \ref{THM:actual utility function} implies that the SEU function becomes 
\begin{equation}\label{EQ:decomposition of SEU}
\mathbb{E}_\mu\bigl(u(f)\bigr)=\sum_{s\in S} \mu(s) w(s) \bar{u}(f_s),
\end{equation}
where $w(s)=\bar{\mu}(s)/\mu(s)$ is the utility from state $s$. So, although state-independent preferences do not identify the actual belief, they still allow us to separately identify utility from a consequence at a state and utility from the state itself. 
\end{remark}

\section{Decision-theoretic foundations}\label{S:decision theoretic foundation}

In this section, I provide axiomatic foundations for the conditions that appear in Definition \ref{DEF:proxy}, in order to clarify which assumptions are exogenous and which are testable with traditional choice data. 

Aligned with \cite{AnscombeAumann1963}, let $Q=\Delta(X)$ be the set of lotteries over a finite set of prizes $X$.\footnote{The analysis can be easily extended to other decision-theoretic frameworks, e.g., \cite{Savage1954} or \cite{Wakker1989}.} Take the extended state space $S\times T$, and denote the set of all acts by $\mathcal{F}:=Q^{S\times T}$, with $f_{s,t}:=f(s,t)$ being the lottery attached to state $(s,t)$ by act $f$. The set of $S$-measurable acts is denoted by $\mathcal{F}_S$. Compound acts are defined in the usual way, i.e., $\lambda f+(1-\lambda)g$ induces the outcome $\lambda f_{s,t}+(1-\lambda)g_{s,t}$ at state $(s,t)\in S\times T$. For any $f,g\in\mathcal{F}$ and any $A\subseteq S\times T$, define the act $f_Ag$ by
$$(f_A g)(s,t)=\begin{cases}
f(s,t) & \mbox{if } (s,t)\in A,\\
g(s,t) & \mbox{if } (s,t)\notin A,
\end{cases}$$
i.e., $f_A g$ coincides with $f$ in $A$ and with $g$ everywhere else. 

The agent has (weak) preferences $\succeq$ over $\mathcal{F}$. As usual, $\succ$ and $\sim$ denote the asymmetric part (viz., strict preference) and the symmetric part (viz., indifference) respectively. Conditional preferences $\succeq_A$ are defined in the usual way: $f\succeq_A g$ if and only if $f_A h\succeq g_A h$ for all $h\in \mathcal{F}$. Throughout the paper, for notational simplicity, we write $\succeq_s:=\succeq_{\{s\}\times T}$ and $\succeq_{s,t}:=\succeq_{\{s\}\times \{t\}}$. Note that the preference relation $\trianglerighteq$  that we considered in Section \ref{S:identification problem} is not necessarily the same as $\succeq$ restricted in $\mathcal{F}_S$. 

An event $A\subseteq S\times T$ is called null if for any two acts $f,g\in\mathcal{F}$ we have $f \sim_A g$. For the time being, to simplify presentation, we assume that there are no null states:\footnote{Without exogenously assuming that there are no null states, the definition of a proxy needs to be slightly adjusted. First of all, observe that the support of $\pi$ can be uniquely identified from $\succeq$. Then, once we have pinned down the set of non-null states, conditions $(P_0)$, $(P_1)$ and $(P_3)$ will be respectively replaced by the following assumptions: (1) the agent has no stakes in $T$ conditional on any non-null $s\in S$, (2) the uninformative event $E\subseteq T$ is non-null, and (3) the vectors in $\{\pi_T(\cdot|s)|s\mbox{ is non-null}\}$ are linearly independent. }
\begin{itemize}
\item[$(A_0)$] \textsc{Full support ---} No state $(s,t)\in S\times T$ is null.
\end{itemize}

Consider the following standard axioms:
\begin{itemize}
\item[$(A_1)$] \textsc{Completeness ---} For all $f,g\in\mathcal{F}:$ $f\succeq g$ or $g\succeq f$.
\item[$(A_2)$] \textsc{Transitivity ---} For all $f,g,h\in\mathcal{F}:$ if $f\succeq g$ and $g\succeq h$, then $f\succeq h$.
\item[$(A_3)$] \textsc{Continuity ---} For all $f,g\in\mathcal{F}:$ $\{g\in\mathcal{F}: f\succeq g\}$ and $\{g\in\mathcal{F}: g\succeq f\}$ are closed in $\mathcal{F}$.
\item[$(A_4)$] \textsc{Independence ---} For all $f,g,h\in\mathcal{F}$, and for all $\lambda\in(0,1):$  $f\succeq g$ if and only if $\lambda f+(1-\lambda)h\succeq \lambda f+(1-\lambda)h$.
\end{itemize}
These axioms form the basic premise for connecting beliefs to choices through some SEU representation \citep{KarniSchmeidlerVind1983}.  Formally, $(A_1)-(A_4)$ are satisfied if and only if there exists a linear (state-dependent) utility function $\tilde{u}:Q\rightarrow\mathbb{R}^{S\times T}$ and a joint probability measure $\tilde{\pi}\in\Delta(S\times T)$ such that 
for all $f,g\in\mathcal{F}$, 
\begin{equation}\label{EQ:vNM state-dependent SEU}
f\succeq g \ \Leftrightarrow \ \mathbb{E}_{\tilde{\pi}}\bigl(\tilde{u}(f)\bigr)\geq \mathbb{E}_{\tilde{\pi}}\bigl(\tilde{u}(g)\bigr),
\end{equation} 
where the SEU from an arbitrary act $f$ is denoted by
\begin{equation}
\mathbb{E}_{\tilde{\pi}}\bigl(\tilde{u}(f)\bigr)=\sum_{s\in S}\sum_{t\in T} \tilde{\pi}(s,t) \tilde{u}_{s,t}(f_{s,t}).
\end{equation}
If we add $(A_0)$ on top of $(A_1)-(A_4)$, the belief $\tilde{\pi}$ becomes full-support. The fact that the utility function is linear implies that there is a (state-dependent) vNM utility function $\tilde{v}:X\rightarrow\mathbb{R}^{S\times T}$ such that, for all lotteries $q\in Q$, 
$$\tilde{u}_{s,t}(q)=\sum_{x\in X} q(x)\tilde{v}_{s,t}(x).$$

I henceforth treat $(A_0)-(A_4)$ as prerequisite assumptions, not because they cannot be tested, but rather because the belief identification problem arises only once they hold. Then, assuming that $(A_0)-(A_4)$ are satisfied, I will proceed to restrict the class of SEU representations that are consistent with $T$ being a proxy for $S$, given that $\pi_T^{\obj}$ has been already labelled as \textit{the objective marginal belief}. The latter does not mean that the agent's actual marginal $\pi_T$ necessarily agrees with $\pi_T^{\obj}$. It only means that there are commonly known objective facts from which $\pi_T^{\obj}$ is obtained.

\begin{definition}
A full-support SEU representation $(\bar{v},\bar{\pi})$ is called \textit{Conditionally State-Independent (abbrev., CSI-SEU)} whenever $\bar{v}$ is $S$-measurable. Furthermore, a CSI-SEU representation $(\bar{v},\bar{\pi})$ is called \textit{Proxy-Consistent (abbrev., PC-SEU)}, whenever (a) $\bar{\pi}_T=\pi_T^{\obj}$, and (b) $\bar{\pi}_T(\cdot|s_1),\dots,\bar{\pi}_T(\cdot|s_K)$ are linearly independent.
\end{definition}

Being able to represent $\succeq$ with a PC-SEU $(\bar{v},\bar{\pi})$ is the minimal necessary condition for $T$ to be a proxy, viz., $S$-measurability of $\bar{v}$ is necessary for $(P_0)$; agreement of $\bar{\pi}_T$ with $\pi_T^{\obj}$ is necessary for $(P_2)$; and linear independence of $\bar{\pi}_T(\cdot|s_1),\dots,\bar{\pi}_T(\cdot|s_K)$ is necessary for $(P_3)$. However, existence of a PC-SEU representation is not sufficient by any means, because $(P_0)-(P_3)$ involve exogenous assumptions. For instance, even if we establish that $(\bar{v},\bar{\pi})$ is a PC-SEU representation, it is never possible to conclude that $\bar{\pi}$ is the actual joint belief $\pi$, for the same reasons we could not conclude $\bar{\mu}=\mu$. In this sense, we can only reject the hypothesis that $T$ is a proxy by showing that $\succeq$ does not have a PC-SEU representation, but we cannot verify it even if a PC-SEU representation exists. I further discuss this issue at the end of this section. For the time being, I will axiomatize the class of PC-SEU representations, beginning with the following axiom:
\begin{itemize}
\item[$(A_5)$] \textsc{Local state-monotonicity ---} For all $s\in S$, for all $t,t'\in T$, and for all $p,q\in Q:$ $p\succeq_{s,t} q$ if and only if $p\succeq_{s,t'} q$.
\end{itemize}
This is simply a weakening of the standard state-monotonicity axiom \citep{AnscombeAumann1963}, in that it postulates state-independent preferences only across states in $\{s\}\times T$, rather than across all states in $S\times T$. It is not difficult to show that, together with $(A_0)-(A_4)$, it axiomatizes CSI-SEU. But most importantly, analogously to \cite{AnscombeAumann1963}, it guarantees that the conditional belief $\bar{\pi}_T(\cdot|s)$ is uniquely identified for every $s\in S$, i.e., all CSI-SEU representations yield the same $\bar{\pi}_T(\cdot|s_1),\dots,\bar{\pi}_T(\cdot|s_K)$ (Lemma \ref{L:Exogenous no stakes}).

In order to formally state my next two axioms, I first need to introduce some additional notation. Let $\mathcal{E}=\{ f^{\textup{high}}, f^{\textup{low}}\}$ be a binary menu of strict-dominance-ordered $S$-measurable acts, i.e., for all states $(s,t)$, it is the case that $f^{\high}_{s,t}=:f^{\high}_s$ and $f^{\low}_{s,t}=:f^{\low}_s$, and moreover
\begin{equation*}
f^{\textup{high}}\succ_{s,t} f^{\textup{low}}.
\end{equation*}
For instance, if $\succeq$ satisfies $(A_5)$, one such menu is formed if we take $f^{\high}$ and $f^{\low}$ to be a $\succeq$-maximal and a $\succeq$-minimal act respectively. If such menu exists, consider the lattice (with respect to the dominance relation),
\begin{equation}\label{EQ:lattice}
\mathcal{F}_\mathcal{E}:=\{f\in\mathcal{F}:f^{\high}\succeq_{s,t} f\succeq_{s,t} f^{\low} \mbox{ for all } (s,t)\in S\times T\}.
\end{equation}
Unlike $\mathcal{E}$ itself, $\mathcal{F}_\mathcal{E}$ does not necessarily contain only $S$-measurable acts. For instance, in our earlier example (where $f^{\high}$ and $f^{\low}$ are $\succeq$-maximal and $\succeq$-minimal acts respectively), $\mathcal{F}_\mathcal{E}$ is the entire set $\mathcal{F}$. Since at least one of the two relations in (\ref{EQ:lattice}) is strict, by continuity, for every $f\in\mathcal{F}_\mathcal{E}$ and every state $(s,t)$ there exists a unique $\lambda_{s,t}^{\mathcal{E},f}\in[0,1]$ such that 
\begin{equation}
f\sim_{s,t} (1-\lambda_{s,t}^{\mathcal{E},f})f^{\low}+\lambda_{s,t}^{\mathcal{E},f} f^{\high}.
\end{equation}
Thus, each $f\in\mathcal{F}_\mathcal{E}$ is associated with a unique vector $\lambda^{\mathcal{E},f}\in[0,1]^{S\times T}$. Then, define
\begin{equation}
\mathcal{T}_\mathcal{E}=\{f\in\mathcal{F}:\lambda^{\mathcal{E},f} \mbox{ is } T \mbox{-measurable}\}. 
\end{equation}
In general, $T$-measurability of $\lambda^{\mathcal{E},f}$ neither implies nor is it implied by $T$-measurability of $f$.

Finally, for every $f\in\mathcal{F}$, define the $S$-measurable act 
\begin{equation}
f_s^{\obj}:=\sum_{t\in T}\pi_T^{\obj}(t)f_{s,t}.
\end{equation}
The idea is that $f^{\obj}$ averages across the lotteries that $f$ assigns to the different states in $\{s\}\times T$ with respect to the exogenously given objective marginal $\pi_T^{\obj}$, and subsequently assigns to each state $(s,t)$ the average lottery $f_s^{\obj}$. 

Then, I am ready to introduce the two new axioms:

\begin{itemize}
\item[$(A_6)$] \textsc{Objective mixture indifference ---} There is some binary strict-dominance-ordered $\mathcal{E}\subseteq\mathcal{F}_S$ such that for all $f\in\mathcal{T}_\mathcal{E}:$ $f^{\obj}\sim f$.
\item[$(A_7)$] \textsc{Unique extension ---} For any binary strict-dominance-ordered $\mathcal{E}\subseteq\mathcal{F}_S$, and for any $\succeq'$ satisfying $(A_0)-(A_5):$
if $\succeq \ = \ \succeq'$ in $\mathcal{T}_\mathcal{E}$ then $\succeq \ = \ \succeq'$ in $\mathcal{F}_\mathcal{E}$.
\end{itemize}

Axiom $(A_6)$ will guarantee that, there is a CSI-SEU representation $(\bar{v},\bar{\pi})$ such that $\bar{\pi}_T=\pi_T^{\obj}$ (Lemma \ref{L:Exogenous Objective Marginal}). This axiom is of somewhat non-standard nature as it is defined with respect to the exogenously given $\pi_T^{\obj}$. It essentially postulates that if we replace every conditional belief $\bar{\pi}_T(\cdot|s)$ with the objective belief $\pi_T^{\obj}$, the preferences over $\mathcal{T}_\mathcal{E}$ will not change. 

Axiom $(A_7)$ will guarantee that the uniquely identified beliefs $\bar{\pi}_T(\cdot|s_1),\dots,\bar{\pi}_T(\cdot|s_K)$ that have been previously obtained as part of every CSI-SEU representation are linearly independent (Lemma \ref{L:Test linear independence}). 

\begin{theorem}\label{T:Axiomatization}
There exists a PC-SEU representation $(\bar{v},\bar{\pi})$ if and only if $\succeq$ satisfies $(A_0)-(A_7)$. Furthermore, $\bar{\pi}$ is uniquely identified from $\succeq$.
\end{theorem}

The first important implication of the previous result is that it allows us to formally state the exogenous assumptions in the definition of a proxy. These assumptions are: (1) the objective marginal belief $\pi_T^{\obj}$ agrees with the actual marginal belief $\pi_T$, (2) the conditional beliefs $\bar{\pi}_T(\cdot|s_1),\dots,\bar{\pi}_T(\cdot|s_K)$ agree with the actual conditional beliefs $\pi_T(\cdot|s_1),\dots,\pi_T(\cdot|s_K)$, and (3) the conditional belief $\bar{\pi}_S(\cdot|E)$ agrees with the actual belief $\mu$. Not surprisingly, all three of them have to do with the relationship between $\bar{\pi}$ and the actual beliefs. 

The second important implication of Theorem \ref{T:Axiomatization} goes back to the long-standing debate on the definition of subjective probability \citep[e.g.,][and references therein]{GrantZandt2008, Karni2014}. Although I personally subscribe to the more modern approach (according to which actual beliefs exist), this paper can be easily rewritten along the lines of the classical view (according to which the only observable primitive is the preference relation, and therefore beliefs should be simply defined as the probability measure that we uniquely identify within a class of SEU representations). Indeed, using Theorem \ref{T:Axiomatization}, I could have simply labelled the uniquely identified $\bar{\pi}_S(\cdot|E)$ as the agent's belief, without going into the whole discussion on whether this is her actual beliefs, or even on whether an actual belief exists in the first place.

\section{Conclusion}\label{S:conclusion}

In this paper, I proposed a novel approach for identifying subjective beliefs without exogenously imposing the awkward assumption of state-independent utilities. The idea is to enlarge the state space by introducing a second dimension, which I call a proxy. The key property of a proxy is that the agent has no stakes in its realization conditional on the original state space. It is exactly this property that allows us to uniquely identify the agent's conditional beliefs about the proxy given each realization of the original state space, using a variant of the strategy method. This method is less data-demanding and more flexible compared to other methods in the literature. Thus, it is not just theoretically sound, but also tractable. Of course, there are various important details which pertain to the experimental implementation of the methodology and need to be taken into account before the method is put in practice (e.g., calibrating for updating biases, accounting for the possibility of uncertain beliefs, dealing with the complexity that comes from explicitly spelling out the incentives of the elicitation mechanism or from the need to perform contingent reasoning, etc). Nevertheless, addressing all these issues is outside the scope of the present paper, as they are all orthogonal to the fundamental long-standing belief identification problem that this paper tackles.

	\appendix
	
	\numberwithin{equation}{section}
	
	\setcounter{lemma}{0}
	\setcounter{corollary}{0}
	\setcounter{theorem}{0}
	\setcounter{proposition}{0}
	\setcounter{rema}{0}
	\setcounter{defin}{0}
	
	\renewcommand{\thelemma}{\Alph{section}\arabic{lemma}}
	\renewcommand{\theproposition}{\Alph{section}\arabic{proposition}}
	\renewcommand{\thecorollary}{\Alph{section}\arabic{corollary}}
	\renewcommand{\thetheorem}{\Alph{section}\arabic{theorem}}
	\renewcommand{\therema}{\Alph{section}\arabic{rema}}
	\renewcommand{\thedefin}{\Alph{section}\arabic{defin}}

\section{Proofs}\label{S:proofs}

\begin{proof}[\textup{\textbf{Proof of Theorem \ref{T:Main Theorem}}}]
By $(P_0)$ all conditional beliefs $\pi_T(\cdot|s_1),\dots,\pi_T(\cdot|s_K)$ are uniquely identified by any standard incentive-compatible elicitation task (e.g., a binarized scoring rule) applied $K$ times, once for each $s\in S$. Since $\pi_T$ is in the convex hull of $\{\pi_T(\cdot|s_1),\dots,\pi_T(\cdot|s_K)\}$, there is $\lambda=(\lambda_1,\dots,\lambda_K)\in\mathbb{R}_+^K$ with $\lambda_1+\dots+\lambda_K=1$ such that 
\begin{equation}\label{EQ:proof of main Lemma 1}
\pi_T=\sum_{k=1}^K \lambda_k \pi_T(\cdot|s_k),
\end{equation}
where, by $(P_1)$, we know $\pi_T$. By the law of total probability, $\pi_S$ solves this system. Moreover, it is the unique solution if and only if matrix
$$\Pi=\begin{bmatrix} 
\pi_T(t_1|s_1) &  \cdots & \pi_T(t_N|s_1)\\
\vdots &  \ddots & \vdots\\
\pi_T(t_1|s_K) &  \cdots & \pi_T(t_N|s_K)
\end{bmatrix}$$
is such that $\rank(\Pi)=K$. Of course, the latter holds if and only if $(P_3)$ is satisfied. Then, we obtain $\pi(s,t)=\pi_S(s)\pi_T(t|s)$ by the chain rule of probability. Finally, by the definition of conditional probability we have 
\begin{equation}
\pi_S(s|E)=\frac{\pi(\{s\}\times E)}{\pi_T(E)}.
\end{equation}
And by $(P_2)$, we obtain $\mu=\pi_S(\cdot|E)$, which completes the proof.
\end{proof}

\vspace{1\baselineskip}

\begin{proof}[\textup{\textbf{Proof of Theorem \ref{THM:actual utility function}}}]
First, we will prove that, if $u$ is defined as in (\ref{EQ:identification of actual utility function}) then $(u,\mu)$ is a SEU representation. Take arbitrary $(\alpha_1,\dots,\alpha_K)\in\mathbb{R}^K$ such that $\sum_{k=1}^K\mu(s_k)\alpha_k=\alpha$, and $\beta>0$. Then, for all $f\in\mathcal{F}_S$,  
\begin{equation}
\mathbb{E}_\mu(u(f))=\alpha+\beta\mathbb{E}_{\tilde{\mu}}(\tilde{u}(f)),
\end{equation}
which completes this part of the proof.

\vspace{0.3\baselineskip} \noindent Let us now prove the converse, i.e., if $(\tilde{u},\mu)$ is a SEU then $\tilde{u}$ is necessarily defined as in (\ref{EQ:identification of actual utility function}). For each $k=1,\dots,K$, define the (convex) range of the original state-utility function, $Y_k:=u_k(Q)\subseteq\mathbb{R}$. Then, it is obvious that, the state-utility function $\tilde{u}_k$ is obtained via a strictly increasing transformation of $u_k$, i.e., there is a continuous strictly increasing $\phi_k:Y_k\rightarrow \mathbb{R}$ such that 
\begin{equation}\label{EQ:proof Thm2 transformation}
\tilde{u}_k=\phi_k\circ u_k. 
\end{equation}
This is because both $u_k$ and $\tilde{u}_k$ represent the same state-preferences. 

We will now show that $\phi_k$ is linear. Take an arbitrary $y\in\inter(Y_1\times\cdots\times Y_K)$, and for each $k=1,\dots,K$ define $y^k\in\mathbb{R}^K$ by 
\begin{equation}\label{EQ:definition yk}
y_\ell^k=\begin{cases}
y_\ell & \mbox{if } \ell\neq k,\\
y_k+\nicefrac{\delta}{\mu(s_k)} & \mbox{if } \ell=k.
\end{cases}
\end{equation}
Since $y$ is an interior point, there exists some $\delta_y>0$ such that $y^k\in Y_1\times\cdots\times Y_K$ for all $\delta\in(0,\delta_y)$, and every $k=1,\dots,K$. Hence, for each $k=1,\dots,K$, there exists an act $f^k\in\mathcal{F}_S$ such that, for every $\ell=1,\dots,K$, it is the case that
\begin{equation}\label{EQ:proof Thm2 state utilities}
u_\ell(f_\ell^k)=y_\ell^k, 
\end{equation}
Therefore, by construction, we obtain
\begin{equation}\label{EQ:proof Thm2 SEU 1}
\mathbb{E}_\mu(u(f^k))=\sum_{\ell=1}^K \mu(s_\ell)y_\ell^k=\sum_{\ell=1}^K \mu(s_\ell)y_\ell+\delta.
\end{equation}
Since the right hand-side does not depend on $k$, we have $f^1\sim\cdots\sim f^K$. Then, since $(\tilde{u},\mu)$ is a SEU representation, we obtain
\begin{equation}\label{EQ:proof Thm 2 indifference}
\mathbb{E}_\mu(\tilde{u}(f^1)=\cdots=\mathbb{E}_\mu(\tilde{u}(f^K).
\end{equation}
By combining (\ref{EQ:proof Thm2 transformation}) and (\ref{EQ:proof Thm2 state utilities}), we get
\begin{equation}
\mathbb{E}_\mu(\tilde{u}(f^k)=\sum_{\ell=1}^K \mu(s_\ell)\phi_\ell(y_\ell^k),
\end{equation}
and consequently, by (\ref{EQ:proof Thm 2 indifference}), it follows that for any two distinct $k,\ell=1,\dots,K$,
$$\sum_{m=1}^K\mu(s_m)\phi_k(y_m^k)=\sum_{m=1}^K\mu(s_m)\phi_\ell(y_m^\ell).$$
Using the definition of $y^k$ and $y^\ell$ from (\ref{EQ:definition yk}), the previous equation can be rewritten as
$$\mu(s_k)\phi_k(y_k+\nicefrac{\delta}{\mu(s_k)})+\mu(s_\ell)\phi_\ell(y_\ell)=\mu(s_k)\phi_k(y_k)+\mu(s_\ell)\phi_\ell(y_\ell+\nicefrac{\delta}{\mu(s_\ell)}).$$
Rearranging terms, dividing both sides with $\delta$, and taking right limits, yields
\begin{equation}\label{EQ:left derivatives}
\phi_{k+}'(y_k)=\lim_{\delta\downarrow0}\frac{\phi_k(y_k+\nicefrac{\delta}{\mu(s_k)})-\phi_k(y_k)}{\nicefrac{\delta}{\mu(s_k)}}=\lim_{\delta\downarrow0}\frac{\phi_\ell(y_\ell+\nicefrac{\delta}{\mu(s_\ell)})-\phi_\ell(y_\ell)}{\nicefrac{\delta}{\mu(s_\ell)}}=\phi_{\ell+}'(y_\ell).
\end{equation}
Note that the right derivatives are well-defined as the respective domains $Y_k$ and $Y_\ell$ are convex sets. We repeat this exercise with any $y'\in\inter(Y_1\times\cdots\times Y_K)$ which agrees with $y$ at all coordinates except $k$, i.e., $y_k\neq y_k'$ and $y_\ell= y_\ell'$ for all $\ell\neq k$. Thus, we obtain 
\begin{equation}\label{EQ:left derivatives1}
\phi_{k+}'(y_k')=\phi_{\ell+}'(y_\ell').
\end{equation}
But, since $y_\ell=y_\ell'$, it follows directly that, for any two $y_k,y_k'\in\inter(Y_k)$,
\begin{equation}\label{EQ:left derivatives2}
\phi_{k+}'(y_k)=\phi_{k+}'(y_k'),
\end{equation}
i.e., the right derivative of $\phi_k$ is constant in the interior of its domain. Therefore, together with continuity (including at the boundaries in case those belong to $Y_k$) it implies that $\phi_k$ is linear for all $k=1,\dots,K$. But then, by (\ref{EQ:left derivatives}), the slope of $\phi_k$ is the same for all $k=1,\dots,K$, which completes the proof.
\end{proof}

\vspace{1\baselineskip}

\begin{lemma}\label{L:Exogenous no stakes}
There is a CSI-SEU representation $(\bar{v},\bar{\pi})$ if and only if $\succeq$ satisfies $(A_0)-(A_5)$. Moreover, the conditional beliefs $\bar{\pi}_T(\cdot|s_1),\dots,\bar{\pi}_T(\cdot|s_K)$ are uniquely identified from $\succeq$. 
\end{lemma}

\begin{proof}[\textup{\textbf{Proof}}]
It follows directly from applying \cite{AnscombeAumann1963} to $\succeq_s$ for every $s\in S$.
\end{proof}

\vspace{1\baselineskip}

\begin{lemma}\label{L:Exogenous Objective Marginal}
Suppose that $\succeq$ satisfies $(A_0)-(A_5)$, and let $\bar{\pi}_T(\cdot|s_1),\dots,\bar{\pi}_T(\cdot|s_K)$ be the uniquely identified conditional beliefs from Lemma \ref{L:Exogenous no stakes}. Then, there is some CSI-SEU representation $(\bar{v},\bar{\pi})$ with $\bar{\pi}_T=\pi_T^{\obj}$ if and only if $\succeq$ satisfies $(A_6)$.
\end{lemma} 

\begin{proof}[\textup{\textbf{Proof}}] \textsc{Necessity:} By $(A_0)$ and $(A_3)$, there exists some $\delta>0$, and a pair of lotteries $q_s^{\high},q_s^{\low}\in Q$ for each $s\in S$, such that 
$$\bar{u}_s(q_s^{\high})-\bar{u}_s(q_s^{\low})=\delta.$$ 
Then, take $\mathcal{E}=\{f^{\high},f^{\low}\}$ such that, for every $s\in S$ and $t\in T$, 
\begin{eqnarray*}
f_{s,t}^{\high}&:=&q_s^{\high},\\
f_{s,t}^{\low}&:=&q_s^{\low}.
\end{eqnarray*}
Note that, for every $f\in\mathcal{F}_\mathcal{E}$, 
\begin{equation}\label{EQ: Proof Lemma 2-1}
\bar{u}_s(f_{s,t})=\lambda_{s,t}^{\mathcal{E},f}\delta+\bar{u}_s(q_s^{\low})
\end{equation}
Hence, for all $f\in\mathcal{T}_\mathcal{E}$, by $T$-measurability of $\lambda^{\mathcal{E},f}$, \begin{equation}\label{EQ: Proof Lemma 2-2}
\mathbb{E}_{\bar{\pi}}\bigl(\bar{u}(f)\bigr)=\delta\sum_{t\in T}\bar{\pi}_T(t)\lambda_{t}^{\mathcal{E},f}+\sum_{s\in S}\bar{\pi}_S(s)\bar{u}_s(q_s^{\low}).
\end{equation}
Furthermore, for the same $f\in\mathcal{T}_\mathcal{E}$, we obtain
\begin{eqnarray}
\mathbb{E}_{\bar{\pi}}\bigl(\bar{u}(f^{\obj})\bigr)&=&\sum_{s\in S}\bar{\pi}_S(s)\bar{u}_s(f_s^{\obj})\label{EQ: Proof Lemma 2-3}\\
&=&\sum_{s\in S}\bar{\pi}_S(s)\sum_{t\in T}\bar{\pi}_T(t)\bar{u}_s(f_{s,t})\label{EQ: Proof Lemma 2-4}\\
&=&\delta\sum_{t\in T}\bar{\pi}_T(t)\lambda_{t}^{\mathcal{E},f}+\sum_{s\in S}\bar{\pi}_S(s)\bar{u}_s(q_s^{\low}).\label{EQ: Proof Lemma 2-5}
\end{eqnarray}
where (\ref{EQ: Proof Lemma 2-3}) follows from $S$-measurability of $f^{\obj}$, (\ref{EQ: Proof Lemma 2-4}) follows from the definition of $f^{\obj}$ combined with $\bar{\pi}_T=\pi_T^{\obj}$, and (\ref{EQ: Proof Lemma 2-5}) follows from (\ref{EQ: Proof Lemma 2-1}) combined with $T$-measurability of $\lambda^{\mathcal{E},f}$. Hence, $f\sim f^{\obj}$, i.e., $(A_6)$ holds.

\vspace{0.5\baselineskip} \noindent \textsc{Sufficiency:} Let $\mathcal{E}=\{f^{\high},f^{\low}\}$ be the menu of $S$-measurable acts such that $f\sim f^{\obj}$ for all $f\in\mathcal{F}_\mathcal{E}$. Take an arbitrary CSI-SEU $(\hat{v},\hat{\pi})$ of $\succeq$. Note that there exist $\alpha_s\in\mathbb{R}$ and $\beta_s>0$ such that $\alpha_s+\beta_s \hat{u}_s(f_s^{\high})=1$ and $\alpha_s+\beta_s \hat{u}_s(f_s^{\low})=0$. Then, define the normalized $S$-measurable $\bar{u}$ by
\begin{equation}
\bar{u}_s:=\alpha_s+\beta_s \hat{u}_s
\end{equation}
for each $s\in S$.Moreover, for every $s\in S$, define the new marginal belief 
\begin{equation}
\bar{\pi}_S(s):=\frac{\hat{\pi}_S(s)/\beta_s}{\sum_{s'\in S} \hat{\pi}_S(s')/\beta_{s'}},
\end{equation}
and observe that for every $f\in\mathcal{F}$, we obtain
\begin{equation}
\mathbb{E}_{\hat{\pi}}\bigl(\hat{u}(f)\bigr)=\alpha+\beta\mathbb{E}_{\bar{\pi}}\bigl(\bar{u}(f)\bigr),
\end{equation}
where $\alpha:=\sum_{s\in S}\alpha_s\bar{\pi}_S(s)$ and $\beta:=\sum_{s'\in S}\hat{\pi}_S(s')/\beta_{s'}>0$ are both constants, and the joint belief $\bar{\pi}$ is defined by 
\begin{equation}
\bar{\pi}(s,t):=\bar{\pi}_S(s)\hat{\pi}_T(t|s). 
\end{equation}
Hence, the pair $(\bar{v},\bar{\pi})$ is a CSI-SEU. 

\vspace{0.5\baselineskip} \noindent Fix an arbitrary $t\in T$, and take the act 
$$\langle t\rangle:=f^{\high}_{S\times \{t\}}f^{\low},$$
meaning that for every state $(s,t)$, the state-utility becomes
$$\bar{u}_s(\langle t\rangle_{s,t'})=\begin{cases}
1 & \mbox{ if } t=t',\\
0 & \mbox{ if } t\neq t'.
\end{cases}$$
Hence, by construction, we have
\begin{equation}
\mathbb{E}_{\bar{\pi}}(\bar{u}(\langle t\rangle))=\sum_{s\in S} \bar{\pi}_S(s)\bar{\pi}_T(t|s)=\bar{\pi}_T(t).
\end{equation}
Moreover, by the definition of $\langle t\rangle^{\obj}$, we obtain
\begin{equation}
\mathbb{E}_{\bar{\pi}}(\bar{u}(\langle t\rangle^{\obj}))=\sum_{s\in S} \bar{\pi}_S(s)\pi_T^{\obj}(t)=\pi_T^{\obj}(t).
\end{equation}
Finally, note that by construction, the act $\langle t\rangle$ belongs to $\mathcal{T}_\mathcal{E}$. Hence, by $(A_6)$, we have $\bar{\pi}_T=\pi_T^{\obj}$.
\end{proof}

\vspace{1\baselineskip}

\begin{lemma}\label{L:Test linear independence}
Suppose that $\succeq$ satisfies $(A_0)-(A_5)$, and let $\bar{\pi}_T(\cdot|s_1),\dots,\bar{\pi}_T(\cdot|s_K)$ be the uniquely identified conditional beliefs from Lemma \ref{L:Exogenous no stakes}. Then, $\bar{\pi}_T(\cdot|s_1),\dots,\bar{\pi}_T(\cdot|s_K)$ are linearly independent if and only if $\succeq$ satisfies $(A_7)$.
\end{lemma}

\begin{proof}[\textup{\textbf{Proof}}]
\textsc{Intermediate step:} Fix an arbitrary menu of $S$-measurable acts $\mathcal{E}=\{f^{\high},f^{\low}\}$ such that $f^{\high}\succ_{s,t} f^{\low}$ for all states $(s,t)$. Take a CSI-SEU representation $(\hat{v},\hat{\pi})$ of $\succeq$, and like in the proof of sufficiency in the previous lemma, obtain another CSI-SEU representation $(\bar{v},\bar{\pi})$ such that $\bar{u}_s(f_s^{\high})=1$ and $\bar{u}_s(f_s^{\low})=1$. Hence, for every $f\in\mathcal{F}_\mathcal{E}$, and every state $(s,t)$, we have
$$\bar{u}_s(f_{s,t})=\lambda_{s,t}^{\mathcal{E},f}.$$
Hence, for every $f\in \mathcal{T}_\mathcal{E}$, by $T$-measurability of $\lambda^{\mathcal{E},f}$, we get 
\begin{equation}\label{EQ:Utility of T-measurable act}
\mathbb{E}_{\bar{\pi}}\bigl(\bar{u}(f)\bigr)=\sum_{t\in T}\lambda_t^{\mathcal{E},f}\sum_{s\in S}\bar{\pi}_S(s)\bar{\pi}_T(t|s).
\end{equation}


\vspace{1\baselineskip} \noindent \textsc{Sufficiency:} Suppose that $\bar{\pi}_T(\cdot|s_1),\dots,\bar{\pi}_T(\cdot|s_K)$ are not linearly independent. Hence, there exists some $\bar{\pi}'_S\in\Delta(S)$ with $\bar{\pi}'_S\neq\bar{\pi}_S$ such that
\begin{equation}\label{EQ:linear dependence}
\sum_{s\in S}\bar{\pi}'_S(s)\bar{\pi}_T(\cdot|s)=\sum_{s\in S}\bar{\pi}_S(s)\bar{\pi}_T(\cdot|s)=\bar{\pi}_T.
\end{equation} 
Consider the preference relation $\succeq'$ which is represented by $(\bar{v},\bar{\pi}')$. By $S$-measurability of $\bar{v}$, it follows that $\succeq'$ satisfies $(A_0)-(A_5)$. Moreover, by (\ref{EQ:Utility of T-measurable act}), it follows that $\mathbb{E}_{\bar{\pi}}\bigl(\bar{u}(f)\bigr)=\mathbb{E}_{\bar{\pi}'}\bigl(\bar{u}(f)\bigr)$ for every $f\in\mathcal{T}_\mathcal{E}$, i.e., 
\begin{equation}
\succeq \ = \ \succeq' \mbox{ in }\mathcal{T}_\mathcal{E}.
\end{equation}
For any $\lambda_0,\lambda_1,\lambda_2\in(0,1)$, define the following acts that belong to $\mathcal{F}_\mathcal{E}$:
\begin{eqnarray*}
f^0&:=&\lambda_0 f^{\high}+(1-\lambda_0)f^{\low},\\ 
f^1&:=&\lambda_1 f^{\high}+(1-\lambda_1)f^{\low},\\ 
f^2&:=&\lambda_2 f^{\high}+(1-\lambda_2)f^{\low}.  
\end{eqnarray*}
Then, observe that the following equivalences hold for all $s\in S$:
\begin{eqnarray*}
f^1_{\{s\}\times T}f^2\succeq f^0&\Leftrightarrow&\bar{\pi}_S(s)\lambda_1+(1-\bar{\pi}_S(s))\lambda_2\geq \lambda_0,\\
f^1_{\{s\}\times T}f^2\succeq' f^0&\Leftrightarrow&\bar{\pi}'_S(s)\lambda_1+(1-\bar{\pi}'_S(s))\lambda_2\geq \lambda_0.
\end{eqnarray*}
By $\bar{\pi}_S\neq\bar{\pi}'_S$, there is some $s\in S$ such that $\bar{\pi}_S(s)>\bar{\pi}'_S(s)$. Hence, for some $0<\lambda_2<\lambda_0<\lambda_1<1$, we obtain $f^1_{\{s\}\times T}f^2\succeq f^0$ and $f^1_{\{s\}\times T} f^2\nsucceq' f^0$, meaning that 
\begin{equation}
\succeq \ \neq \ \succeq' \mbox{ in } \mathcal{F}_\mathcal{E}.
\end{equation}
Hence, $\succeq$ is not uniquely extended from $\mathcal{T}_\mathcal{E}$ to $\mathcal{F}_\mathcal{E}$, meaning that $(A_6)$ is violated.

\vspace{1\baselineskip} \noindent \textsc{Necessity:} Suppose that $\bar{\pi}_T(\cdot|s_1),\dots,\bar{\pi}_T(\cdot|s_K)$ are linearly independent. Take an arbitrary $\succeq'$ satisfying $(A_0)-(A_5)$, such that $\succeq \ = \ \succeq'$ in $\mathcal{T}_\mathcal{E}$. This implies that $\succeq_s \ = \ \succeq_s'$ for all $s\in S$, as $\mathcal{T}_\mathcal{E}$ places restrictions only across $T$. Therefore, the CSI-SEU representations $(\bar{v},\bar{\pi})$ and $(\bar{v}',\bar{\pi}')$ that we constructed in the intermediate step, are such that $\bar{v}=\bar{v}'$ and $\bar{\pi}_T(\cdot|s)=\bar{\pi}_T'(\cdot|s)$ for all $s\in S$.  

\vspace{0.5\baselineskip} \noindent Suppose that $\bar{\pi}_S\neq\bar{\pi}_S'$. Then, by linear independence, there exists some $t\in T$ such that
\begin{equation}\label{EQ:linear independence}
\bar{\pi}_T'(t)=\sum_{s\in S}\bar{\pi}_S'(s)\bar{\pi}_T(t|s)<\sum_{s\in S}\bar{\pi}_S(s)\bar{\pi}_T(t|s)=\bar{\pi}_T(t).
\end{equation}
Pick some $\lambda\in(0,1)$ such that $\bar{\pi}'_T(t)<\lambda<\bar{\pi}_T(t)$, and define the act $h$ such that $\bar{u}_{s,t}(h_{s,t})=\lambda$ for every $(s,t)$. It is not difficult to verify that such act exists and $h\in\mathcal{T}_\mathcal{E}$. Moreover, we have $f^{\high}_{S\times \{t\}}f^{\low}\in\mathcal{T}_\mathcal{E}$. Hence, by (\ref{EQ:Utility of T-measurable act}), we obtain 
\begin{equation}
f^{\high}_{S\times \{t\}}f^{\low}\succeq h \mbox{ and } f^{\high}_{S\times \{t\}}f^{\low}\nsucceq' h, 
\end{equation}
which contradict $\succeq \ = \ \succeq'$ in $\mathcal{T}_\mathcal{E}$. Therefore, it must necessarily be the case that $\bar{\pi}_S=\bar{\pi}_S'$. But then, this implies $\bar{\pi}=\bar{\pi}'$, which together with $\bar{v}=\bar{v}'$, yields 
\begin{equation}
\succeq \ = \ \succeq' \mbox{ in } \mathcal{F}_\mathcal{E},
\end{equation}
and the proof is complete.
\end{proof}

\vspace{1\baselineskip}

\begin{proof}[\textup{\textbf{Proof of Theorem \ref{T:Axiomatization}}}]
The proof follows directly from Lemmas \ref{L:Exogenous no stakes}, \ref{L:Exogenous Objective Marginal} and \ref{L:Test linear independence}.
\end{proof}

\begin{small}

\end{small}

\end{document}